\documentclass[journal, twocolumn, 10pt, table]{IEEEtran}

\usepackage{graphicx}
\usepackage{amsmath, amsthm, amssymb, bm}
\usepackage{amsthm}
\usepackage{algpseudocode}
\usepackage{epstopdf}
\usepackage{color}
\usepackage{subfigure}
\usepackage{epsfig}
\usepackage{cleveref}
\usepackage{color}

\newtheorem{theorem}{Theorem}

\newtheorem{remark}{Remark}

\newtheorem{assumption}{Assumption}

\usepackage[linesnumbered,ruled]{algorithm2e}

\newcommand{\skipval}{0.901mm} 
\long\def\comment#1{}

\DeclareMathOperator*{\argmin}{arg\,min}

\newfont{\bbb}{msbm10 scaled 700}

\newfont{\bb}{msbm10 scaled 1100}

\newcommand{\dv}{{\bf d}}

\newcommand{\nv}{{\bf n}}

\newcommand{\xv}{{\bf x}}
\newcommand{\yv}{{\bf y}}

\newcommand{\Dm}{{\bf D}}

\newcommand{\Hm}{{\bf H}}
\newcommand{\Id}{{\bf I}}

\newcommand{\Qm}{{\bf Q}}

\newcommand{\Tm}{{\bf T}}
\newcommand{\Um}{{\bf U}}
\newcommand{\Wm}{{\bf W}}

\newcommand{\epsilonv}{\hbox{\boldmath$\epsilon$}}

\newcommand{\phiv}{\hbox{\boldmath$\phi$}}
\newcommand{\psiv}{\hbox{\boldmath$\psi$}}

\newcommand{\sigmav}{\hbox{\boldmath$\sigma$}}

\newcommand{\Psim}{\hbox{\boldmath$\Psi$}}

\newcommand{\Xim}{\hbox{\boldmath$\Xi$}}

\ifCLASSINFOpdf
\else
\fi
\begin{document}
%

\title{SNR Estimation in Linear Systems with \\ Gaussian Matrices}
%
%
%

%

\author{Mohamed~A.~Suliman, Ayed~M.~Alrashdi, Tarig~Ballal, and Tareq~Y.~Al-Naffouri

\thanks{
\textcopyright 2017 IEEE. Personal use of this material is permitted. Permission from IEEE must be obtained for all other uses, in any current or future media, including reprinting/republishing this material for advertising or promotional purposes, creating new collective works, for resale or redistribution to servers or lists, or reuse of any copyrighted component of this work in other works.

The authors are with the Computer, Electrical, and Mathematical Sciences and Engineering (CEMSE) Division, King Abdullah University of Science and Technology (KAUST), Thuwal, Saudi Arabia. Emails:$\{$mohamed.suliman, ayed.alrashdi, tarig.ahmed, tareq.alnaffouri$\}$@kaust.edu.sa.

A. M. Alrashdi is also with the Department of Electrical Engineering, University of Hail, Saudi Arabia.
}
}

\maketitle

\begin{abstract}
This paper proposes a highly accurate algorithm to estimate the signal-to-noise ratio (SNR) for a linear system from a single realization of the received signal. We assume that the linear system has a Gaussian matrix with one sided left correlation. The unknown entries of the signal and the noise are assumed to be independent and identically distributed with zero mean and can be drawn from any distribution. We use the ridge regression function of this linear model in company with tools and techniques adapted from random matrix theory to achieve, in closed form, accurate estimation of the SNR without prior statistical knowledge on the signal or the noise. Simulation results are provided, and show that the proposed method is very accurate.
\end{abstract}

\IEEEpeerreviewmaketitle
\section{Introduction}
\label{sec:intro}
The need for an accurate estimate of the signal-to-noise ratio (SNR) has attracted a lot of interest in many fields of signal processing and communication such as signal detection and estimation \cite{kailath2000linear, poor2013introduction, 7585088}, machine learning \cite{evgeniou2000regularization, sebald2000support, shalev2014understanding}, image restoration \cite{katsaggelos1991regularized, pesquet2009sure}, cognitive radio \cite{mitola1999cognitive}, MIMO and massive MIMO systems \cite{das2012snr, lu2014overview, 6951471}. This is motivated by the fact that many algorithms in these fields require accurate SNR knowledge to achieve optimal performance.

Over the years, the problem of SNR estimation has been studied in different contexts depending on the application. Historically, it was first inspected in the context of single-input single-output (SISO) systems with constant channels \cite{benedict1967joint}. 

In \cite{das2012snr}, an SNR estimation approach for MIMO systems is proposed based on the known pilot symbols and the unknown data symbols. Both the channel and the noise are assumed to be Gaussian. In addition, the transmitted symbols are assumed to be coming from a known modulation. The authors in \cite{matzner1994snr} use higher-order moments to estimate the SNR assuming that the shape of the narrowband signal and the noise power spectral density are known. In \cite{pauluzzi2000comparison}, several SNR estimation techniques have been compared for AWGN channel, while in \cite{vartiainen2006blind} a blind SNR estimation technique for narrowband signals corrupted by complex Gaussian noise is proposed.

In this paper, we consider the problem of estimating the SNR from a vector $\yv \in \mathbb{R}^{M}$ of noisy linear observations that are obtained through
\begin{equation}
\label{eq:model}
\yv =  \Wm \xv_{\text{o}} +\nv,
\end{equation}
where $\xv_{\text{o}} \in \mathbb{R}^{K}$ is an unknown transmitted signal and $\nv \in \mathbb{R}^{M}$ is an unknown noise vector. A key difference between the proposed estimator and the aforementioned estimators is that it does not impose any assumptions on the distribution functions of $\xv_{\text{o}}$ and $\nv$ as well as their spectra.

In this work, we assume that the entries of $\xv_{\text{o}}$, and those of the noise $\nv$, are drawn from any two distributions and that there elements are independent and identically distributed (i.i.d.) with zero mean and unknown variances. The matrix $\Wm$ is an $M\times K $ matrix that can be decomposed as
\begin{equation}
\label{eq:channel matrix}
\Wm = \Psim^{\frac{1}{2}} \overline{\Wm},
\end{equation}
where $\Psim$ is an $M \times M$ known Hermitian nonnegative left correlation matrix, while $\overline{\Wm} \in \mathbb{R}^{M \times K}$ is a known Gaussian matrix with centered unit variance (standard) i.i.d. entries. 

The optimal signal estimate of the model in (\ref{eq:model}) is obtained through the linear minimum mean-squared error (LMMSE) estimator which is given by \cite{poor2013introduction}
\begin{equation}
\label{eq:lmmse simple}
\hat{\xv} =\left(\Wm^{T}\Wm + \lambda_{\text{o}}\Id_{\text{K}}\right)^{-1}\Wm^{T}\yv,
\end{equation}
where the regularization parameter $\lambda_{\text{o}} =\frac{\sigma_{\nv}^{2}}{\sigma_{\xv_{\text{o}}}^{2}}=\frac{1}{\text{SNR}}$. The expression in (\ref{eq:lmmse simple}) is the regularized least squares (RLS) estimation of $\xv_{\text{o}}$ and it is the minimizer of the regression function
\begin{equation}
\label{eq:tik-minimization}
\hat{\xv} := \underset{\xv}{\operatorname{\argmin}} \big\{ ||\yv - \Wm \xv ||_2^{2} + \lambda \  ||\xv||_{2}^{2}\big\}|_{\lambda = \lambda_{\text{o}}}.
\end{equation}
In this paper, we use the cost function in (\ref{eq:tik-minimization}) to estimate the SNR of the linear system in (\ref{eq:model}) by applying tools from random matrix theory (RMT). We prove that the objective function in (\ref{eq:tik-minimization}), evaluated at its optimal solution $\hat{\xv}$, concentrates at a value that is a function of the system dimensions, the correlation matrix, the regularization parameter, and the signal and the noise variances. This function turns out to be linear in $\sigma_{\xv_{\text{o}}}^{2}$ and $\sigma_{\nv}^{2}$. Thus, by evaluating this function at multiple values of $\lambda$, we can solve for $\sigma_{\xv_{\text{o}}}^{2}$ and $\sigma_{\nv}^{2}$.

The paper is organized as follows. In Section~\ref{sec:main}, we present the main theorem for estimating the SNR of (\ref{eq:model}) and illustrate how it can be applied. Section~\ref{sec:proof} presents the proof of this theorem while Section~\ref{sec:results} illustrates the performance of the proposed SNR estimation algorithm using simulations.

\emph{Notations}-Boldface lower-case symbols are used for column vectors, $\dv$, and upper-case for matrices, $\Dm$. The notation $[\dv]_{i}$ denotes the $i$-th element of $\dv$, while $[\Dm]_{i,j}$ indicates the element in $i$-th row and $j$-th column of $\Dm$. $\Dm^{T}$ and $\text{Tr}\left(\Dm\right)$ denote the transpose and the trace of $\Dm$, respectively.  Moreover, $\Id_{\text{M}}$ denotes the $M \times M$ identity matrix. The expectation operator is denoted by $\mathbb{E} [\cdot]$ while $||\cdot||_{2}$ denotes the spectral norm for matrices and Euclidean norm for vectors. Finally, diag $\left(\dv\right)$ is the diagonal matrix that have the elements of $\dv$ at its diagonal entries, while $``\xrightarrow{\text{a.s.}}"$ designates almost sure convergence.
\section{The SNR Estimation Algorithm}
\label{sec:main}
In this section, we present the main theorem of the paper and we show how it can be applied to estimate the SNR of (\ref{eq:model}). We start by stating the main assumptions of the theorem.
\begin{assumption}
\label{A1}
\normalfont
We consider the linear asymptotic regime in which both $M$ and $K$ go to $+\infty$ with their ratio being bounded below and above as follows:
\begin{equation}
\label{eq:ratio}
0 < \rho^{-} =  \lim \inf \frac{K}{M} \leq \rho^{+} = \lim \sup \frac{K}{M} < +\infty. \nonumber
\end{equation}
\begin{assumption}
\label{A2}
\normalfont
Let $\overline{\Wm} \in \mathbb{R}^{M\times K}$ have i.i.d. entries with $ [ \overline{\Wm} ]_{i,j} \sim \mathcal{N}\left(0, 1\right)$ and define the spectral decomposition of $\Psim$ as $\Psim = \Um \Qm \Um^{T}$. Furthermore, we assume that there exists a real number $q_{\text{max}}< \infty$ such that 
\begin{equation}
\label{eq:d bound}
\sup_{K} || \Qm || \leq q_{\text{max}}, \nonumber
\end{equation}
and that the normalized trace of $\Qm$ satisfies
\begin{equation}
\label{eq:norm trace}
\inf_{K} \frac{1}{K} \text{Tr}\left(\Qm\right) >0. \nonumber
\end{equation}
\end{assumption}
\end{assumption}
\begin{assumption}
\label{A3}
\normalfont
We assume that the entries of $\xv_{\text{o}}$ are sampled i.i.d. from some distribution function (not necessarily known) with zero mean and unknown variance $\sigma_{\xv_{\text{o}}}^{2}$ and that the noise vector $\nv$ also has i.i.d. entries that are sampled from some density function of zero mean and unknown variance $\sigma_{\nv}^{2}$.
\end{assumption}
\begin{remark}\normalfont(On the Assumptions)
\begin{itemize}
\item Although Assumption~\ref{A1} is a key factor in the derivation of the proposed SNR estimation algorithm, we will show later by simulations that the proposed technique has good accuracy even for relatively small system dimensions.
\item The derivation of the proposed algorithm relies on the fact that $\overline{\Wm}$ is a Gaussian matrix.
\item Assumptions~\ref{A1} and \ref{A2} are technical assumptions that are often satisfied in practice. Thus, they do not limit the applicability of the proposed SNR estimation algorithm.
\item Based on Assumption~\ref{A3}, the SNR of (\ref{eq:model}) is $\frac{{\sigma}_{\xv_{\text{o}}}^{2}}{{\sigma}_{\nv}^{2}}$.
\item We confine ourselves to the set of real numbers $\mathbb{R}$, but we point out that the theorem and the results in the paper are directly applicable to the complex set $\mathbb{C}$.
\end{itemize}
\end{remark}
\begin{theorem}
\label{th2}
Under the settings of Assumptions~\ref{A1}, \ref{A2}, and \ref{A3}, and by considering a normalized version of the cost function in (\ref{eq:tik-minimization}) evaluated at its optimal solution $\hat{\xv} \left(\lambda\right)$, i.e.,
\begin{equation}
\label{eq:Expectation of Cost function}
\Phi \left(\overline{\Wm}\right) =  \frac{1}{K} ||\yv - \Psim^{\frac{1}{2}}\overline{\Wm} \hat{\xv} \left(\lambda\right) ||_{2}^{2}+ \frac{\lambda}{K} \ ||\hat{\xv}\left(\lambda\right)||_{2}^{2},
\end{equation}
then, there exists a deterministic function $\alpha\left(t\right)$ defined as
\begin{align}
\label{eq:Main result}
& \alpha\left(t\right) =  
\frac{\text{Tr}\left(\Psim \Tm \left(t\right)\right)}{\left(1+t \delta\left(t\right)\right)}  \sigma_{\xv_{\text{o}}}^{2} +\left( \frac{M}{K} - \frac{ t \ \text{Tr}\left(\Psim \Tm \left(t\right)\right)}{K \left(1+t\delta\left(t\right)\right)} \right) \sigma_{\nv}^{2} \nonumber\\& + \mathcal{O}\left(K^{-1}\right),
\end{align}
such that 
\begin{equation}
\label{eq:almost sure}
\mathop{\mathbb{E}}_{\xv_{\text{o}},\nv}[\Phi \left(\overline{\Wm}\right)] - \alpha\left(t\right) \xrightarrow[]{\text{a.s.}} 0,
\end{equation} 
where $t =\frac{K}{\lambda}$ and the notation $d = \mathcal{O}\left(K^{-1}\right)$ means that $| \frac{d}{K^{-1}}| $ is bounded as $K \to \infty$. The matrix $\Tm\left(t\right)$ is an $M \times M$ matrix that is defined as
\begin{equation}
\label{eq:T maitrx}
\Tm \left(t \right) =  \Um \left( \Id_{\text{M}}  +  \frac{t}{\left(1+t \delta\left(t\right)\right)} \Qm \right)^{-1} \Um^{T},
\end{equation}
where $\delta \left(t\right)$ is the unique positive solution of the following fixed-point equation:
\begin{equation}
\label{eq:delta}
\delta \left(t\right) = \frac{1}{K} \text{Tr}\left( \Qm \left(\Id_{\text{M}} +\frac{t}{1+t \ \delta\left(t\right)} \Qm\right)^{-1} \right).
\end{equation}
\end{theorem}
\begin{proof}
The proof of Thereom~\ref{th2} is given in Section~\ref{sec:proof}.
\end{proof}
\begin{remark}\normalfont(On Theorem~\ref{th2})
The normalization of the cost function in (\ref{eq:tik-minimization}) by $K$ is for RMT purposes. Also note that both the cost function in (\ref{eq:tik-minimization}) and its normalized version have the same minimizer (i.e., same signal estimate $\hat{\xv} \left(\lambda\right)$).
\end{remark}

In Fig.~\ref{fig:de vs exact all}, we evaluate the result of Theorem~\ref{th2} by providing an example that compares $\mathop{\mathbb{E}}_{\xv_{\text{o}},\nv}[\Phi \left(\overline{\Wm}\right)] $ and (\ref{eq:Main result}). We set $\overline{\Wm} \in \mathbb{R}^{300\times 100}, [\overline{\Wm}]_{i,j} \sim \mathcal{N}(0, 1)$, and we take $\Psim^{\frac{1}{2}} = \text{diag}\left(\psiv\right)$ where $[\psiv]_{i}$ is uniformly distributed in $[0,1]$. Finally, we choose $[\xv_{\text{o}}]_{i}\sim \mathcal{N}(0, 10)$ and $[\nv]_{i}\sim \mathcal{N}(0, 1)$. 

From Fig.~\ref{fig:de vs exact all}, we can observe that the result of Theorem~\ref{th2} is very accurate and that the error is negligible. In fact, the same behavior can be observed for different scenarios of the correlation matrix $\Psim$, the signal $\xv_{\text{o}}$, and the noise $\nv$. 
\begin{figure}[h!]
	 \centering	 	                                                                             
	{\includegraphics[width=2.275in]{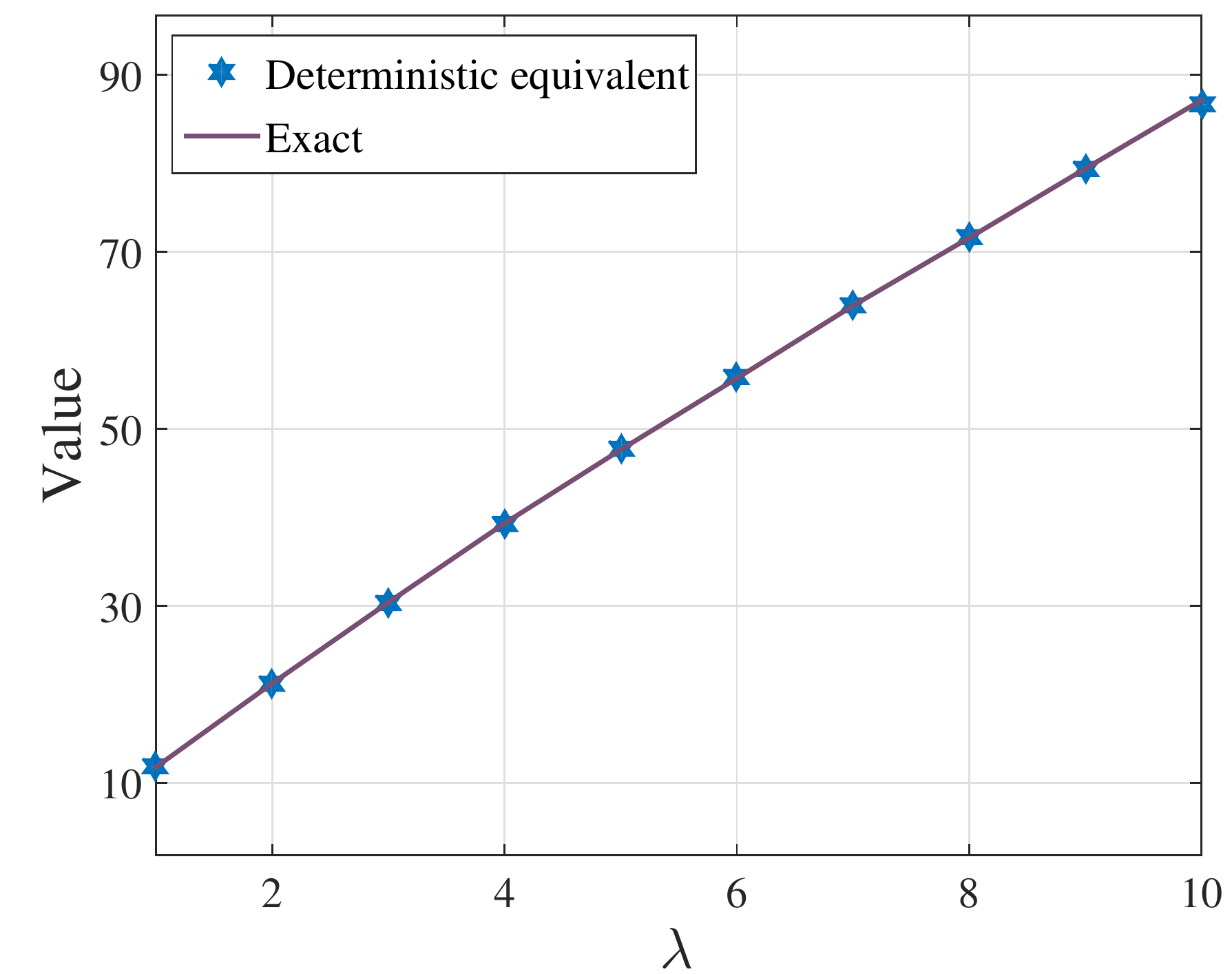}}  
\caption{Comparing $\mathop{\mathbb{E}}_{\xv_{\text{o}},\nv}[\Phi \left(\overline{\Wm}\right)]$ with its deterministic equivalent in (\ref{eq:Main result}).}
	\label{fig:de vs exact all}
\end{figure}

The result in Theorem~\ref{th2} indicates that the average value of the normalized version of the function in (\ref{eq:tik-minimization}) converges to (\ref{eq:Main result}) at $\xv = \hat{\xv}$. Let us see how this can be used to estimate $\sigma_{\xv_{\text{o}}}^{2}$ and $\sigma_{\nv}^{2}$. To this end, set
\begin{equation}
\label{eq:term1}
\xi_{1} \left(\lambda\right) = \frac{\text{Tr}\left(\Psim \Tm \left(t\right)\right)}{\left(1+t \delta\left(t\right)\right)},
\end{equation}
\begin{equation}
\label{eq:term2}
\xi_{2} \left(\lambda\right) = \frac{M}{K} - \frac{ t \ \text{Tr}\left(\Psim \Tm \left(t\right)\right)}{K \left(1+t\delta\left(t\right)\right)},
\end{equation}
and consider the definition of the function $\Phi$ as in (\ref{eq:Expectation of Cost function}). By evaluating (\ref{eq:term1}), (\ref{eq:term2}), and (\ref{eq:Expectation of Cost function}) at different $\lambda$ values, i.e., $\lambda_{i}, i = 1,\dots, n; n\geq 2$, we can generate a system of linear equations in the form
\begin{equation}
\label{eq:matrix formulation 2}
{\begin{bmatrix}  \xi_{1}\left(\lambda_{1}\right)  & \xi_{2}\left(\lambda_{1}\right) \\ \vdots  & \vdots \\ \xi_{1}\left(\lambda_{n}\right)  & \xi_{2}\left(\lambda_{n}\right)  \end{bmatrix}}  \begin{bmatrix} \sigma_{\xv_{\text{o}}}^{2} \\  \sigma_{\nv}^{2} \end{bmatrix} + \begin{bmatrix} \epsilon_{1} \\ \vdots \\ \epsilon_{n} \end{bmatrix} = \begin{bmatrix} \Phi\left(\lambda_{1}\right) \\ \vdots \\ \Phi\left(\lambda_{n}\right) \end{bmatrix}\Rightarrow \Xim \sigmav + \epsilonv = \phiv,
\end{equation}
where $\epsilonv$ is the approximation error vector. Now, by solving the following constrained linear LS problem:
\begin{eqnarray}
\label{eq3:contrained LS} 
 \underset{\sigmav}\min \  \frac{1}{2} \ ||\phiv - \Xim \sigmav  ||_{2}^{2}  \hspace{10pt} \text{subject to} \,\, \sigmav \geq \bm{0},
\end{eqnarray}
we can obtain the values of $\hat{\sigma}_{\xv_{\text{o}}}^{2}$ and $\hat{\sigma}_{\nv}^{2}$.

The error vector $\epsilonv$ in (\ref{eq:matrix formulation 2}) is due to the fact that we are equating the normalized version of the cost function in (\ref{eq:tik-minimization}) with its expected value taken over $\xv_{\text{o}}$ and $\nv$ when they are both evaluated at $\xv =\hat{\xv}$. This will be accurate for high dimensions but becomes less so as we decrease the dimensions of the problem. Moreover, (\ref{eq3:contrained LS}) assumes that the error variables are uncorrelated which is not the case as the values of $\Phi \left(\lambda\right)$ for different $\lambda$ are correlated. Thus, $\Xim \sigmav$ in (\ref{eq3:contrained LS}) should be weighted by the inverse of the unknown covariance matrix of $\epsilonv$. Our main goal here is to find a pair $\left(\hat{\sigma}_{\xv_{\text{o}}}^{2},\hat{\sigma}_{\nv}^{2}\right)$ that closely approximates the signal and the noise statistics for all (or almost all) possible realizations of the random quantities in (\ref{eq:model}). 

The process of estimating the SNR of (\ref{eq:model}) is summarized in Algorithm~\ref{al:algorithm rmt}.
\footnote{The MATLAB code of the proposed SNR estimation approach is provided at http://faculty.kfupm.edu.sa/ee/naffouri/publications.html}
\begin{algorithm}
    \SetKwInOut{Input}{Input}
    \SetKwInOut{Output}{Output}

    \Input{$\yv, \overline{\Wm}, \Psim=\Um \Qm \Um^{T},[\lambda_{1},\dots, \lambda_{n}];  n\geq 2$.}
    \Output{SNR.}
     \For{$i=1,2,\dots n$}
     {
     $t_{i} = \frac{K}{\lambda_{i}}$\;
     Compute $\delta\left(t_{i}\right)$ from (\ref{eq:delta}) and $\Tm\left(t_{i}\right)$ using (\ref{eq:T maitrx})\;
      $\hat{\xv}_{i} =\left(\Wm^{T}\Wm+\lambda_{i}\Id_{\text{K}}\right)^{-1}\Wm^{T}\yv  $\;
      Obtain $\xi_{1} \left(\lambda_{i}\right), \xi_{2} \left(\lambda_{i}\right)$, and $\Phi\left(\lambda_{i}\right)$ using (\ref{eq:term1}), (\ref{eq:term2}), and (\ref{eq:Expectation of Cost function}), respectively.
     
     }
     Fromulate the matrix $\Xim$ and the vector $\phiv$ as in (\ref{eq:matrix formulation 2})\;
     Solve (\ref{eq3:contrained LS}) to obtain $\hat{\sigma}_{\xv_{\text{o}}}^{2}$ and $\hat{\sigma}_{\nv}^{2}$\;
     SNR = $\frac{\hat{\sigma}_{\xv_{\text{o}}}^{2}}{\hat{\sigma}_{\nv_{\text{o}}}^{2}}$.
    \caption{Algorithm to estimate the SNR}
    \label{al:algorithm rmt}
\end{algorithm}
\section{Proof Outline}
\label{sec:proof}
In this section, we provide the proof of Theorem~\ref{th2}. We start by evaluating the function in (\ref{eq:tik-minimization}) at its optimal solution in (\ref{eq:lmmse simple}) for a general $\lambda$. Substituting (\ref{eq:lmmse simple}) in (\ref{eq:tik-minimization}) and manipulating, we obtain
\begin{align}
\label{eq: Cost function at estimate}
&||\yv -  \Psim^{\frac{1}{2}} \overline{\Wm} \hat{\xv}||_{2}^{2}+ \lambda \ ||\hat{\xv}||_{2}^{2}= \yv^{T}\yv + \yv^{T}\Wm \Hm\Wm^{T} \Wm \Hm\Wm^{T}\yv \nonumber\\
&- 2 \yv^{T}\Wm \Hm \Wm^{T}\yv + \lambda \yv^{T}  \Wm \Hm^{2} \Wm^{T}  \yv,
\end{align}
where $\Hm = \left(\Wm^{T}\Wm +\lambda\Id_{\text{K}}\right)^{-1}$.

Now, we consider taking the expected value of (\ref{eq:Expectation of Cost function}) over all its random variables. Based on Assumption~\ref{A3}, we can express the expected value of the first term in (\ref{eq:Expectation of Cost function}) using (\ref{eq:model}) and (\ref{eq: Cost function at estimate}) as
\begin{align}
\label{eq:de of first term}  
&\mathop{\mathbb{E}}_{\overline{\Wm}, \xv_{\text{o}}, \nv} \big[ \frac{1}{K} ||\yv -  \Psim^{\frac{1}{2}} \overline{\Wm} \hat{\xv}||_{2}^{2} \big]= \frac{\sigma_{\xv_{\text{o}}}^{2}}{K} \ \mathop{\mathbb{E}}_{\overline{\Wm}} \Big[ \text{Tr} \left( \Wm^{T}  \Wm\right) \nonumber\\&- 2 \ \text{Tr} \left(\Wm^{T} \Wm \Hm \Wm^{T} \Wm\right)+  \text{Tr} \left(\Wm^{T} \Wm \Hm \Wm^{T}  \Wm \Hm \Wm^{T}  \Wm \right)\Big] \nonumber\\&+  \frac{\sigma_{\nv}^{2}}{K} \ \mathop{\mathbb{E}}_{\overline{\Wm}} \Big[ M + \text{Tr} \left(\Wm^{T} \Wm \Hm\Wm^{T} \Wm \Hm \right) - 2 \text{Tr} \left( \Wm^{T} \Wm \Hm\right) \Big].
\end{align}
\begin{figure*}[ht!]  
	 \centering	 	                                                                             
	\subfigure[var(ML)=0.2358, var(proposed)  = 0.1348.]{\label{fig:out1}\includegraphics[width=2.352in, height=1.7in]{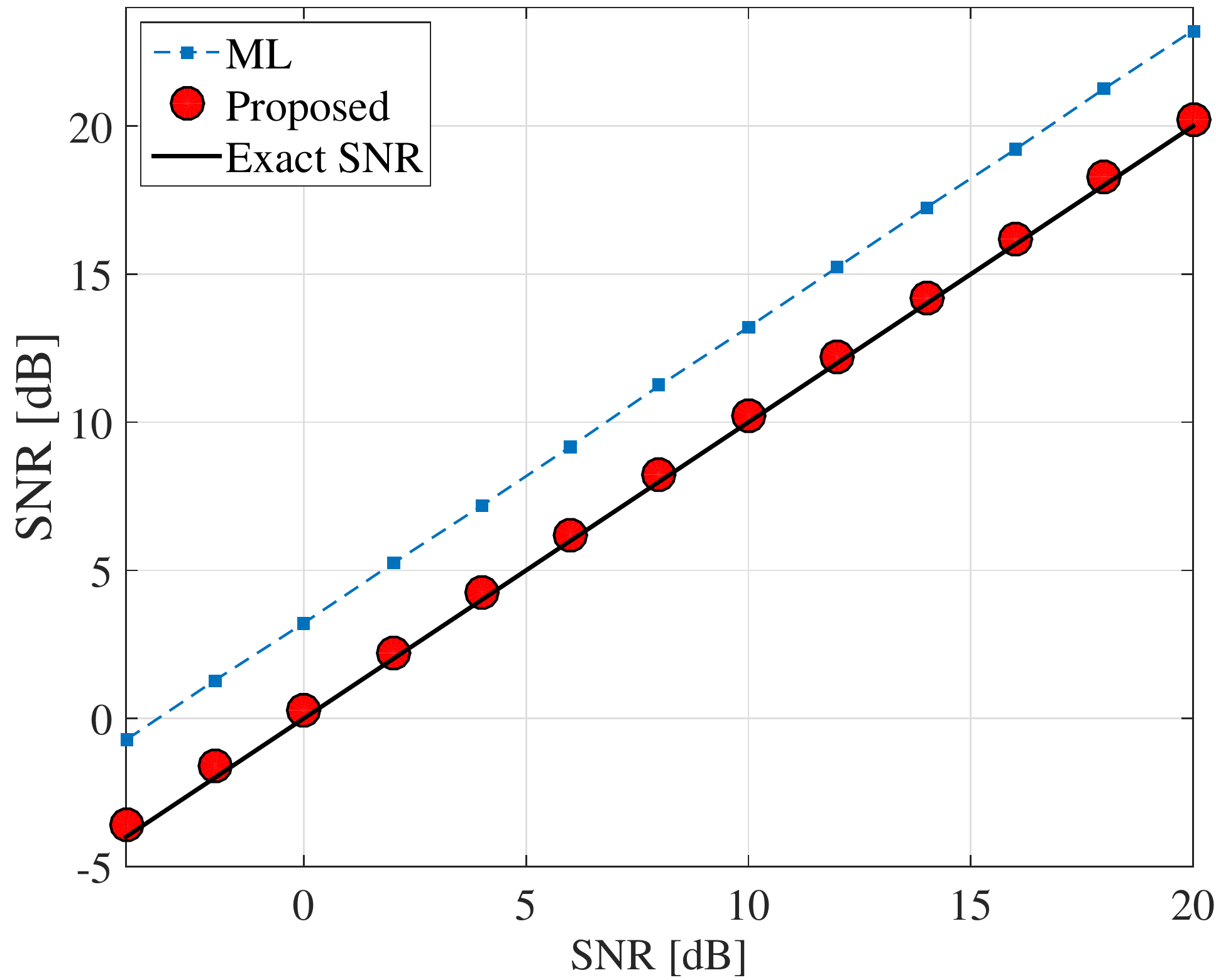}}
	\subfigure[var(ML) =0.1721, var(proposed) =0.1001.]{\label{fig:out2}\includegraphics[width=2.351in, height=1.7in]{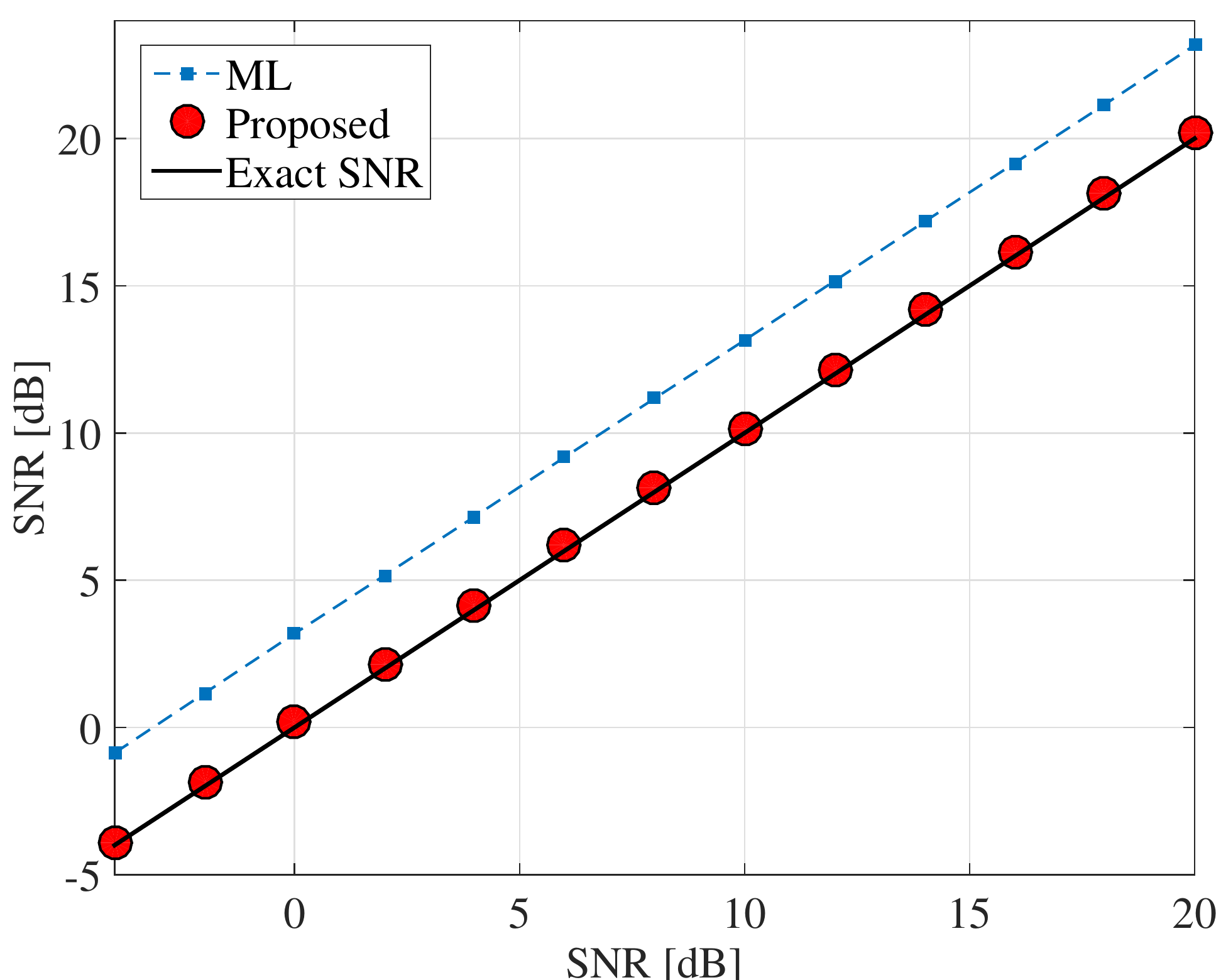}}
	\subfigure[var(ML) = 0.2482, var(proposed)=0.0908.]{\label{fig:out3}\includegraphics[width=2.351in, height=1.7in]{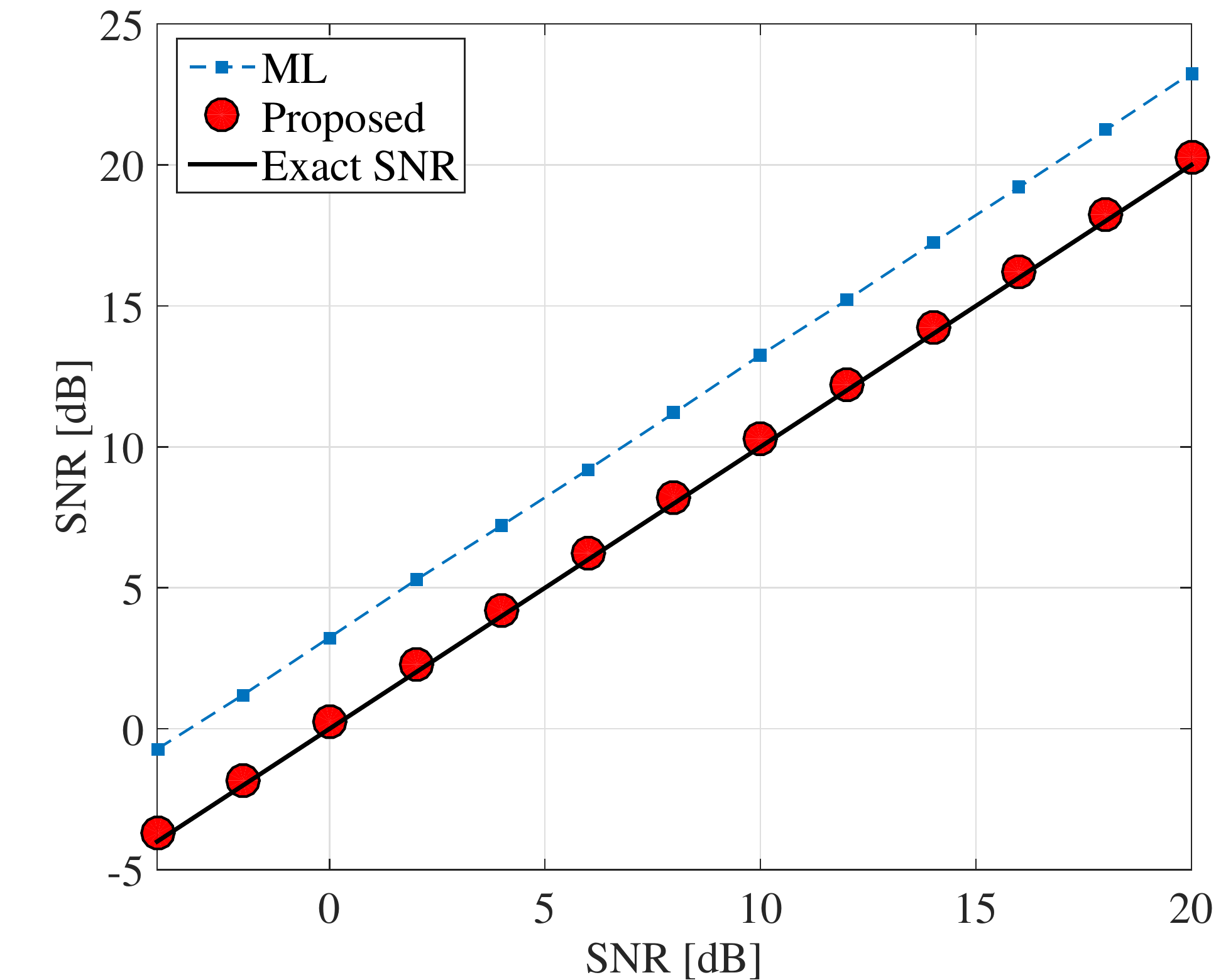}}

	\subfigure[ var(ML) = 0.2426, var(proposed)=  0.0656.]{\label{fig:out4}\includegraphics[width=2.351in, height=1.7in]{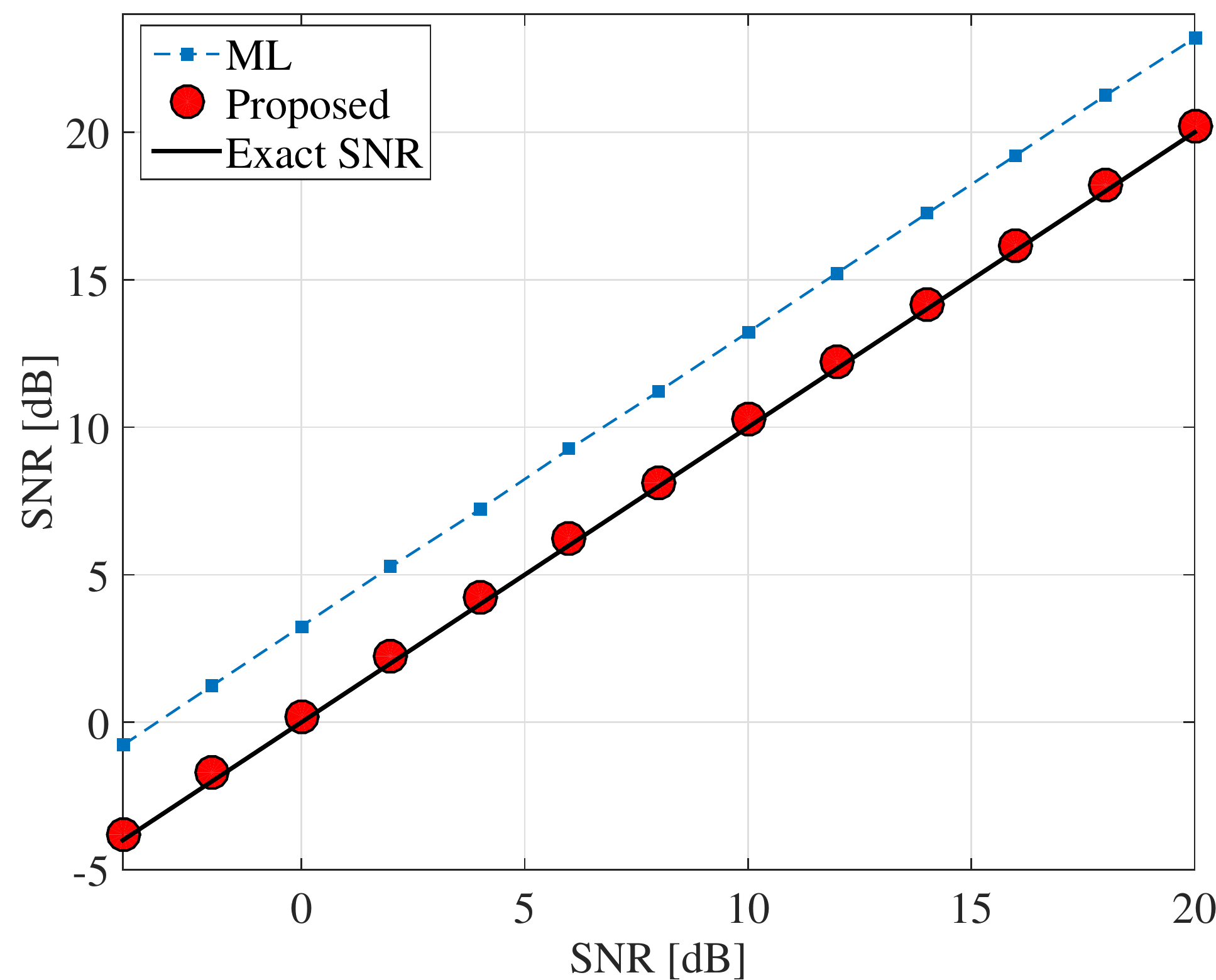}}
	\subfigure[NMSE performance of the algorithms.]{\label{fig:out5}\includegraphics[width=2.351in, height=1.7in]{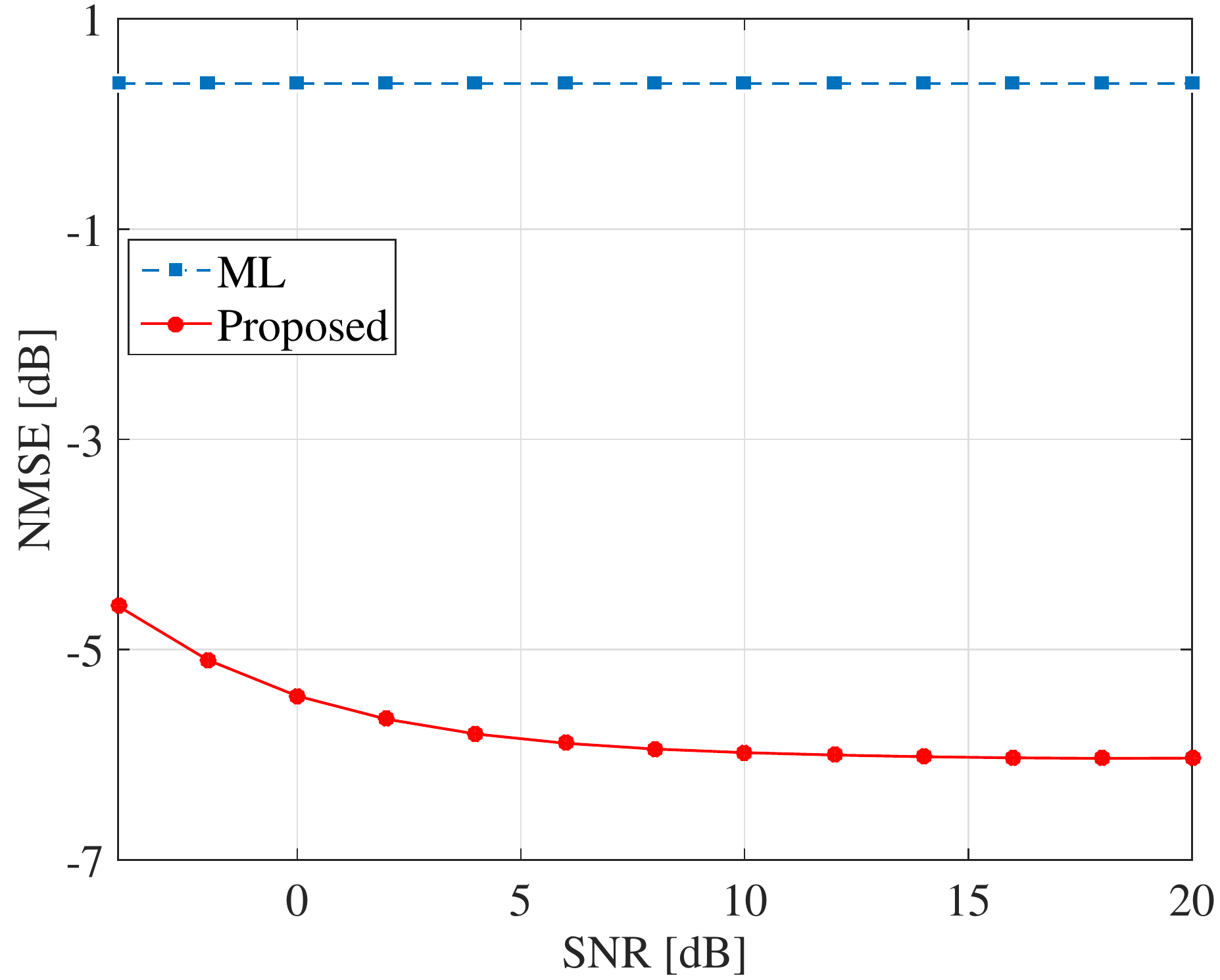}}   
			\subfigure[Example of SNR estimates obtained in 1000 trials per SNR.]{\label{fig:out6}\includegraphics[width=2.351in, height=1.7in]{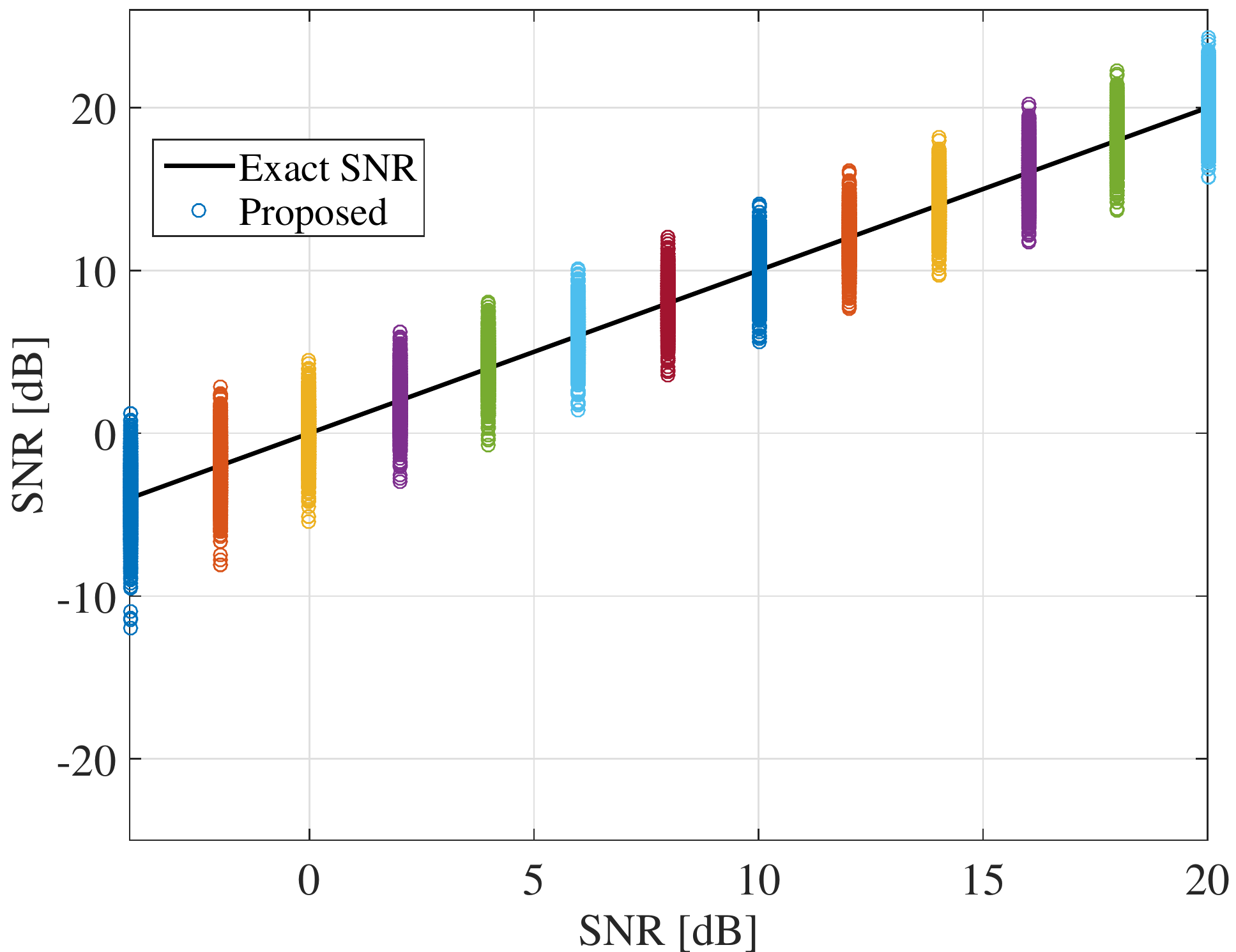}}  
	
	\subfigure[Variation of the performance with system dimensions.]{\label{fig:out7}\includegraphics[width=2.351in, height=1.7in]{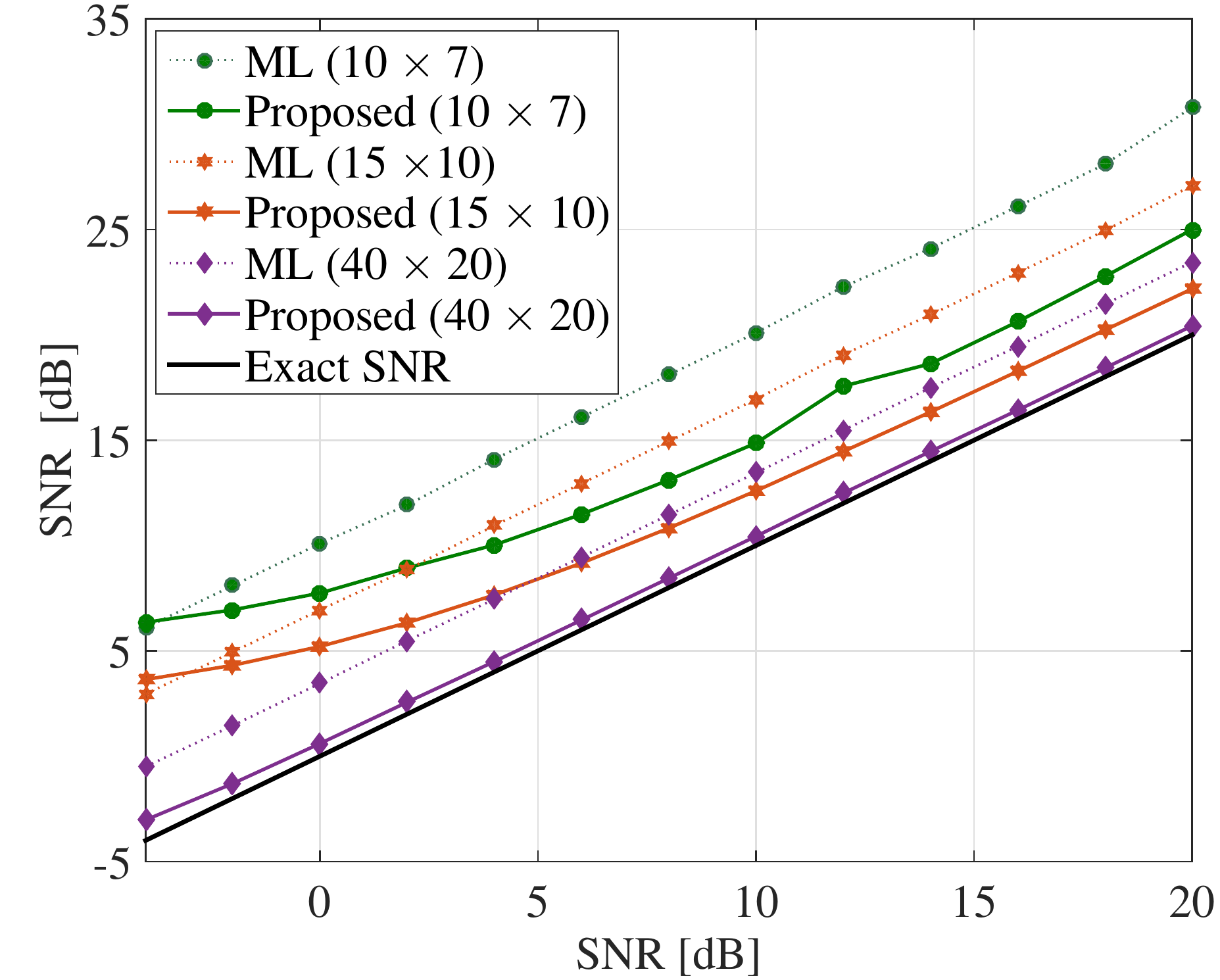}}  
		\subfigure[Variation of the performance with system dimensions.]{\label{fig:out8}\includegraphics[width=2.351in, height=1.7in]{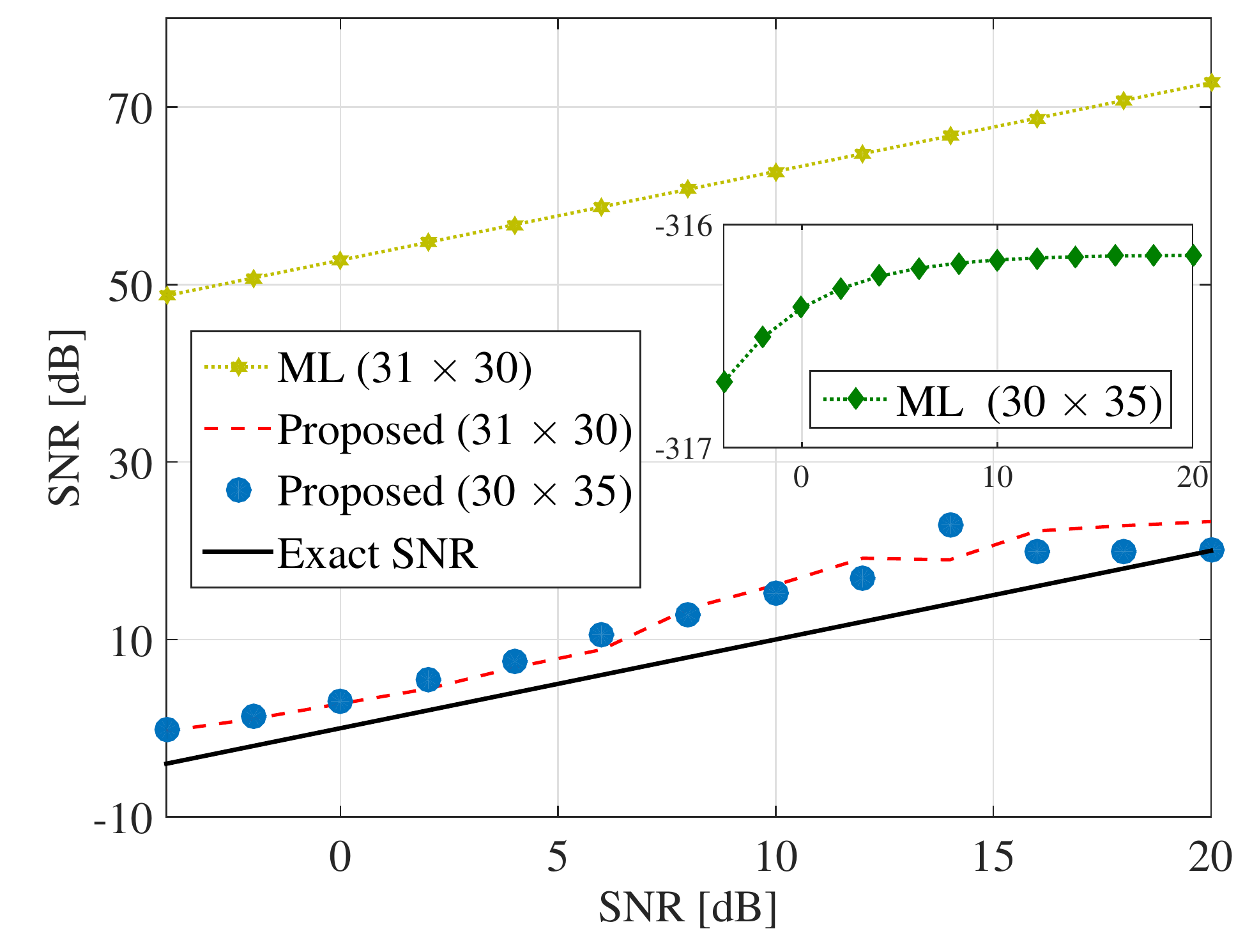}}   %
	\subfigure[Sensitivity of the performance to the choice of $\lambda$.]{\label{fig:out9}\includegraphics[width=2.351in, height=1.7in]{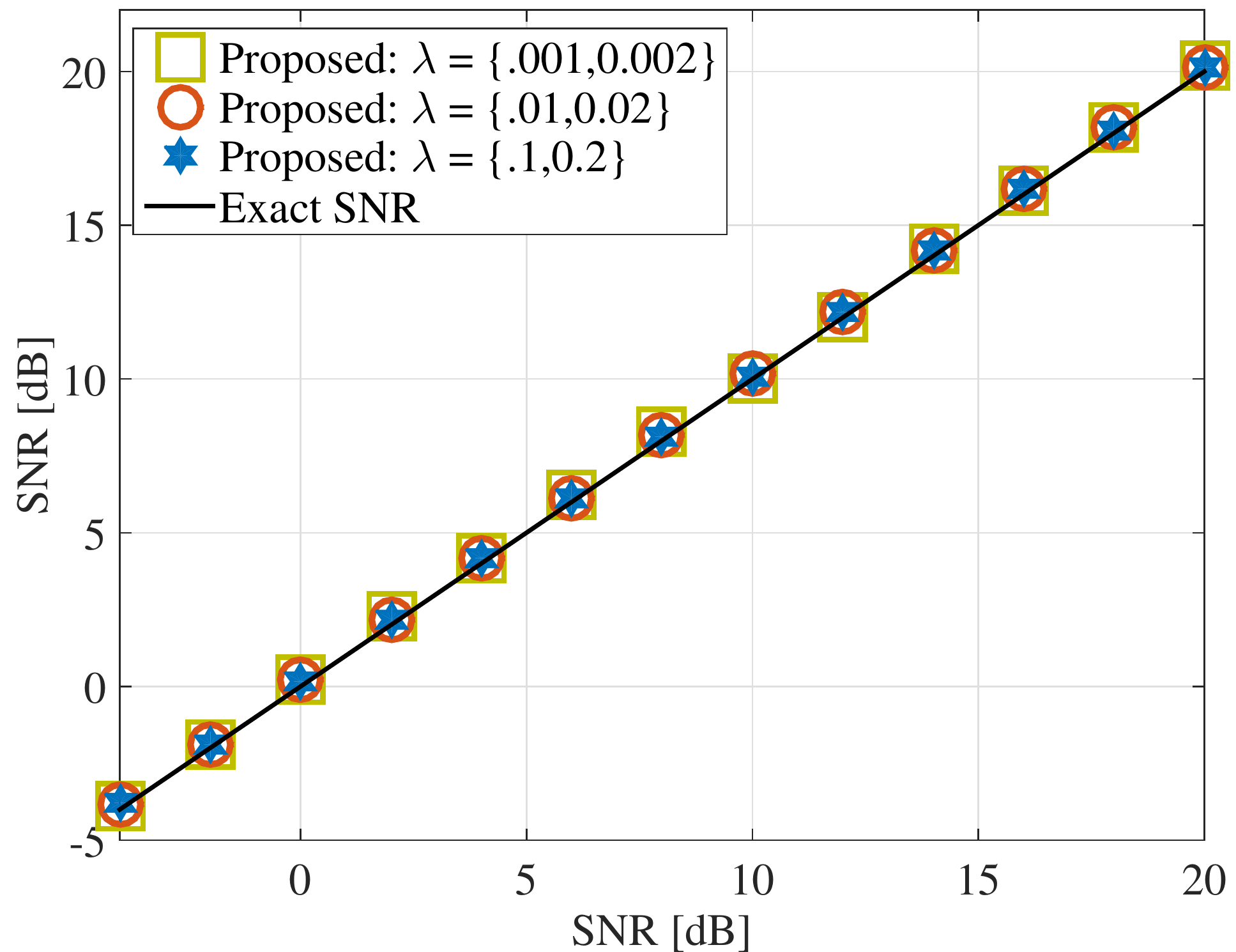}}
					
	\caption{Comparing the performance of the proposed approach in the different scenarios.}
	\label{fig:scenario 1,2,3}
\end{figure*}

After some algebraic manipulations, we can rewrite (\ref{eq:de of first term}) as 
\begin{align}  
\label{eq:eq:de of first term 2}
&\mathop{\mathbb{E}}_{\overline{\Wm}, \xv_{\text{o}}, \nv} \big[ \frac{1}{K}  ||\yv - \Psim^{\frac{1}{2}} \overline{\Wm} \hat{\xv}||_{2}^{2} \big]=\nonumber\\
&\frac{\sigma_{\xv_{\text{o}}}^{2} \lambda^{2}}{K}\ \mathop{\mathbb{E}}_{\overline{\Wm}} \Big[  \text{Tr} \left(\Wm^{T} \Wm \left(\Wm^{T}  \Wm +\lambda \Id_{\text{K}}\right)^{-2}\right)\Big] \nonumber\\
&+\frac{\sigma_{\nv}^{2}}{K} \mathop{\mathbb{E}}_{\overline{\Wm}} \Big[ M - \text{Tr}\left(\Wm^{T}\Wm\left(\Wm^{T}  \Wm +\lambda \Id_{\text{K}}\right)^{-1} \right) \nonumber\\
&- \lambda \text{Tr}\left(\Wm^{T} \Wm\left(\Wm^{T}  \Wm +\lambda \Id_{\text{K}}\right)^{-2} \right)\Big].
\end{align}
Now, by using the following equality
\begin{equation}
\label{eq:identity}
\left(\Wm^{T} \Wm + \Id_{\text{K}}\right)^{-1} \Wm = \Wm \left(\Wm \Wm^{T} +\Id_{\text{M}}\right)^{-1},
\end{equation}
we can further simplify (\ref{eq:eq:de of first term 2}) to 
\begin{align}
\label{eq:de of first term 3}
&\mathop{\mathbb{E}}_{\overline{\Wm}, \xv_{\text{o}}, \nv} \big[ \frac{1}{K} ||\yv - \Psim^{\frac{1}{2}} \overline{\Wm}\hat{\xv}||_{2}^{2} \big]= \nonumber\\
&\frac{\sigma_{\nv}^{2}}{K} \ \mathop{\mathbb{E}}_{\overline{\Wm}} \Bigg[M  - t \ \text{Tr}\left(\left(\frac{t}{K}\Wm \Wm^{T} + \Id_{\text{M}}\right)^{-1} \frac{\Wm \Wm^{T}}{K} \right) \nonumber\\
&- t \ \text{Tr}\left(\left(\frac{t}{K}\Wm \Wm^{T} + \Id_{\text{M}}\right)^{-2} \frac{\Wm \Wm^{T}}{K} \right)\Bigg] \nonumber\\
& +{\sigma_{\xv_{\text{o}}}^{2}} \ \mathop{\mathbb{E}}_{\overline{\Wm}} \Bigg[  \text{Tr} \left(\frac{\Wm \Wm^{T}}{K} \left(\frac{t}{K} \Wm \Wm^{T} + \Id_{\text{M}}\right)^{-2}\right)\Bigg].
\end{align}
Following the same procedure, we can prove that the expected value of the second term in (\ref{eq:Expectation of Cost function}) can be expressed as 
\begin{align}
\label{eq:DE term 2}
&\mathop{\mathbb{E}}_{\overline{\Wm}, \xv_{\text{o}}, \nv} \  [ \frac{\lambda} {K} \ ||\hat{\xv}||_{2}^{2}\big] =  {} \mathop{\mathbb{E}}_{\overline{\Wm}} \Bigg[  \text{Tr}\left(\left(\frac{t}{K}\Wm \Wm^{T} + \Id_{\text{M}}\right)^{-1} \frac{\Wm \Wm^{T}}{K} \right)  \nonumber\\
& - \text{Tr}\left(\left(\frac{t}{K}\Wm \Wm^{T} + \Id_{\text{M}}\right)^{-2} \frac{\Wm \Wm^{T}}{K} \right)\Bigg] \sigma_{\xv_{\text{o}}}^{2} \nonumber\\
& + \mathop{\mathbb{E}}_{\overline{\Wm}} \Bigg[ t\ \text{Tr}\left(\left(\frac{t}{K}\Wm \Wm^{T} + \Id_{\text{M}}\right)^{-2} \frac{\Wm \Wm^{T}}{K} \right)\Bigg] \frac{\sigma_{\nv}^{2}}{K}.
\end{align}
Combining (\ref{eq:de of first term 3}) and (\ref{eq:DE term 2}) yields
\begin{align}
\label{eq: final expression cost function}
&{\mathop{\mathbb{E}}}_{\overline{\Wm}, \xv_{\text{o}}, \nv} \Big[ \frac{1}{K} ||\yv -\Psim^{\frac{1}{2}} \overline{\Wm}\hat{\xv}||_{2}^{2} + \frac{\lambda}{K} \ ||\hat{\xv}||_{2}^{2} \Big]= \frac{M}{K} \sigma_{\nv}^{2}  \nonumber\\
&+ \left( K\sigma_{\xv_{\text{o}}}^{2} - t \sigma_{\nv}^{2}\right) \mathop{\mathbb{E}}_{\overline{\Wm}} \Bigg[\frac{1}{K}\text{Tr}\left(\left(\frac{t}{K}\Wm \Wm^{T} + \Id_{\text{M}}\right)^{-1} \frac{\Wm \Wm^{T}}{K} \right) \Bigg].
\end{align}
Now, based on the result obtained in \cite{4608971} (Equation (23)), and by using some algebraic manipulations, we can prove that for the second term in (\ref{eq: final expression cost function}), the following equality holds true
\begin{align}
\label{eq:DE term 2 term1}
&\mathop{\mathbb{E}}_{\overline{\Wm}} \Bigg[ \frac{1}{K} \text{Tr}\left(\left(\frac{t}{K}\Wm \Wm^{T} + \Id_{\text{M}}\right)^{-1} \frac{\Wm \Wm^{T}}{K} \right)\Bigg]= \frac{1}{K} \frac{\text{Tr}\left(\Psim \Tm\left(t\right)\right)}{\left(1+t\delta\left(t\right)\right)} \nonumber\\
&+\mathcal{O}\left(K^{-2}\right).
\end{align}
By substituting (\ref{eq:DE term 2 term1}) in (\ref{eq: final expression cost function}) and manipulating, we obtain (\ref{eq:Main result}). Now, by using (\ref{eq: final expression cost function}), (\ref{eq:DE term 2 term1}), and (\ref{eq:Main result}), we can conclude that
\begin{equation}
\label{eq:almost sure 2}
{\mathbb{E}}_{\overline{\Wm}, \xv_{\text{o}}, \nv} [\Phi \left(\overline{\Wm}\right)] - \alpha\left(t\right) \xrightarrow[]{} 0.
\end{equation}
However, as shown in \cite{4608971} (Proposition~4), and based on (\ref{eq: final expression cost function}), we have
\begin{equation}
\label{eq: variance of the term}
\text{var}\left(\frac{1}{K}\text{Tr}\left(\left(\frac{t}{K}\Wm \Wm^{T} + \Id_{\text{M}}\right)^{-1} \frac{\Wm \Wm^{T}}{K} \right)\right) = \mathcal{O}\left(K^{-2}\right).
\end{equation}
Now, by using (\ref{eq:almost sure 2}) and (\ref{eq: variance of the term}), and upon applying the Borel-Cantelli lemma \cite{klenke2013probability}, we can easily prove that 
\begin{equation}
\label{eq:almost sure 3}
{\mathbb{E}}_{\xv_{\text{o}}, \nv} [\Phi \left(\overline{\Wm}\right)] - {\mathbb{E}}_{\overline{\Wm}, \xv_{\text{o}}, \nv} [\Phi \left(\overline{\Wm}\right)] \xrightarrow[]{a.s.} 0.
\end{equation}
Finally, based on (\ref{eq:almost sure 2}) and (\ref{eq:almost sure 3}), we can conclude that $\mathop{\mathbb{E}}_{\xv_{\text{o}},\nv}[\Phi \left(\overline{\Wm}\right)]$ converges almost surely to $\alpha\left(t\right)$ as in (\ref{eq:almost sure}).
\section{Simulation Results}
\label{sec:results}
In this section, we evaluate the performance of the proposed SNR estimation method using seven different scenarios that differ in the choices of $\Psim, \xv_{\text{o}}, \nv, \lambda , M$, and $K$. In all the experiments, $\overline{\Wm}$ is generated to satisfy Assumption~\ref{A2}. In the first four scenarios, we set the matrix dimensions to be $80 \times 40$ and we set $\lambda_{i} \in \{1,2,3,4\} \times 10^{-3}$, while in the last three scenarios, we demonstrate the performance of the proposed approach with small matrix dimensions and with different choices of $\lambda$. The performance of the proposed approach is compared with the maximum likelihood (ML) estimator \cite{kay2013fundamentals}, where we use the ML estimator to estimate $\hat{\sigma}_{\nv}^{2}$ and then the SNR is obtained assuming that we know $\sigma_{\xv\text{o}}^{2}$ exactly. The results in each scenario are averaged over $10^{3}$ Monte-Carlo trials. In addition, we compute the average variance of the normalized errors (errors normalized by the true SNR) over the given SNR range and we quote it in the sub-figures captions.
\subsubsection{Scenario (a)}
\label{sub:1}
We choose $\Psim^{\frac{1}{2}} = \text{diag}\left(\psiv\right)$ with $[\psiv]_{i}$ being uniformly distributed in the interval $[0,1]$ i.e., $[\psiv]_{i} \sim \mathcal{U}\left(0,1\right)$. The noise vector $\nv$ is generated such that $[\nv]_{i}\sim \mathcal{N}(0,0.1)$. The entries of $\xv_{\text{o}}$ are set to satisfy $[\xv_{\text{o}}]_{i}\sim \mathcal{N}(0,\sigma_{\xv_{\text{o}}}^{2})$, while $\sigma_{\xv_{\text{o}}}^{2}$ is varied such that the model in (\ref{eq:model}) will have SNR values $\{-4,-2,0,\dots ,20\}$ dB.
\subsubsection{Scenario (b)}
\label{sub:2}
In this scenario, we generate $\Psim$ as
\begin{equation}
\label{eq:correlation matrix 1}
[\Psim]_{i,j} = J_{0}\left(\pi|i-j|^{2}\right),
\end{equation}
where $J_{0}\left(\cdot\right)$ is the zero-order Bessel function of the first kind. Such matrix is used to model the correlation between transmit antennas in a dense scattering environment \cite{alfano2004capacity}. We set $[\nv]_{i} \sim \mathcal{U}\left(-3,3\right)$, while $\xv_{\text{o}}$ is generated as in Scenario~(a).
\subsubsection{Scenario (c)}
\label{sub:3}
We consider the exponential correlation model for $\Psim$ which is defined as \cite{shin2006capacity} 
\begin{equation}
\label{eq:correlation matrix 2}
\Psim \left(\hat{\rho}\right) = \Big[\hat{\rho}^{|i-j|^{2}}\Big]_{i,j=1,2,\dots,M} , \ \ \hat{\rho} \in[0,1),
\end{equation}
with $\hat{\rho} =0.4$. We set $[\xv_{\text{o}}]_{i}\sim \mathcal{U}\left(-5,5\right)$ and $[\nv]_{i} \sim \mathcal{N}\left(0,\sigma_{\nv}^{2}\right)$. The value of $\sigma_{\nv}^{2}$ is varied to simulate different SNR levels. 
\subsubsection{Scenario (d)}
\label{sub:4}
We use the same model for $\Psim$ and $\nv$ as in Scenario~(c). The entries of $\xv_{\text{o}}$ are drawn from a student's t-distribution with a degree of freedom $\nu = 5$ and thus $\sigma_{\xv_{\text{o}}}^{2}=\frac{5}{3}$.

Figs.~\ref{fig:scenario 1,2,3}(a)-(d) show perfect matching between the exact SNR and the estimated SNR using Theorem \ref{th2} in all the four scenarios. It can be observed that the ML estimator tends to underestimate the noise variance in all the scenarios with different amount of error. Moreover, the plots also show that the proposed estimator is unbiased. Finally, the captions of the sub-figures show that the average variance of the proposed approach error is very small and is less than that of the ML.

In Fig.~\ref{fig:out5}, we plot the normalized mean-squared error (NMSE) (in dB) (i.e., MSE normalized by the exact SNR) of the algorithms. From Fig.~\ref{fig:out5}, we can see that the proposed approach has very low normalized error over all the SNR range while the ML exhibits very high NMSE that is above 0 dB. Further, Fig.~\ref{fig:out6} plots the SNR estimates over all the Monte-Carlo trials for Scenario~(a) and shows that the trials are concentrated around the true SNR. 
\subsubsection{Scenarios (g) and (h)} 
\label{sub:1} In these scenarios, we depict the variation of the performance of the proposed approach with matrix dimensions using the settings in Scenario~(a). From Fig.~\ref{fig:out7}, we observe that for very small dimensions, i.e., $10 \times 7$, the proposed approach offers relatively poor performance that still outperforms the ML algorithm. As the dimensions increase, performance enhances. However, even with relatively small dimensions such as $40 \times 20$, the proposed approach maintains its high accuracy with tenuous error. 

On the other hand, Fig.~\ref{fig:out8} presents the performance of two special cases: when $M = K+1$ (using $31 \times 30$) and when $M < K$ (using $30 \times 35$). Fig.~\ref{fig:out8} shows that the proposed approach maintains its high accuracy with small error while the ML exhibits very poor performance in the two cases.
\subsubsection{Scenario (i)}
In this scenario, we study the sensitivity of the proposed approach to the choice of $\lambda$. In Fig.~\ref{fig:out9}, we plot the performance of the proposed approach for three different choices of $\lambda$ using the same settings in Scenario~(b). From Fig.~\ref{fig:out9}, we can observe that for all the different choices of $\lambda$, the algorithm still provides high SNR estimation accuracy.

\section{Conclusions}
\label{sec:conc}
In this paper, we developed a new SNR estimation technique for linear systems with correlated Gaussian channel by using tools from RMT. The proposed approach is shown to provide high SNR estimation accuracy for different scenarios. Moreover, the algorithm maintains its high accuracy even for small matrix dimensions.

\bibliographystyle{IEEEbib}
\bibliography{refs}

\begin{thebibliography}{10}

\bibitem{kailath2000linear}
Thomas Kailath, Ali~H Sayed, and Babak Hassibi,
\newblock {\em Linear estimation}, vol.~1,
\newblock Prentice Hall Upper Saddle River, NJ, 2000.

\bibitem{poor2013introduction}
H~Vincent Poor,
\newblock {\em An introduction to signal detection and estimation},
\newblock Springer Science \& Business Media, 2013.

\bibitem{7585088}
M.~Suliman, T.~Ballal, A.~Kammoun, and T.~Y. Al-Naffouri,
\newblock ``Constrained perturbation regularization approach for signal
  estimation using random matrix theory,''
\newblock {\em IEEE Signal Processing Letters}, vol. 23, no. 12, pp.
  1727--1731, Dec 2016.

\bibitem{evgeniou2000regularization}
Theodoros Evgeniou, Massimiliano Pontil, and Tomaso Poggio,
\newblock ``Regularization networks and support vector machines,''
\newblock {\em Advances in computational mathematics}, vol. 13, no. 1, pp.
  1--50, 2000.

\bibitem{sebald2000support}
Daniel~J Sebald and James~A Bucklew,
\newblock ``Support vector machine techniques for nonlinear equalization,''
\newblock {\em IEEE Transactions on Signal Processing}, vol. 48, no. 11, pp.
  3217--3226, 2000.

\bibitem{shalev2014understanding}
Shai Shalev-Shwartz and Shai Ben-David,
\newblock {\em Understanding machine learning: From theory to algorithms},
\newblock Cambridge university press, 2014.

\bibitem{katsaggelos1991regularized}
Aggelos~K Katsaggelos, Jan Biemond, Ronald~W Schafer, and Russell~M Mersereau,
\newblock ``A regularized iterative image restoration algorithm,''
\newblock {\em IEEE Transactions on Signal Processing}, vol. 39, no. 4, pp.
  914--929, 1991.

\bibitem{pesquet2009sure}
Jean-Christophe Pesquet, Amel Benazza-Benyahia, and Caroline Chaux,
\newblock ``A sure approach for digital signal/image deconvolution problems,''
\newblock {\em IEEE Transactions on Signal Processing}, vol. 57, no. 12, pp.
  4616--4632, 2009.

\bibitem{mitola1999cognitive}
Joseph Mitola and Gerald~Q Maguire,
\newblock ``Cognitive radio: making software radios more personal,''
\newblock {\em IEEE personal communications}, vol. 6, no. 4, pp. 13--18, 1999.

\bibitem{das2012snr}
Aniruddha Das and Bhaskar~D Rao,
\newblock ``Snr and noise variance estimation for mimo systems,''
\newblock {\em IEEE Transactions on Signal processing}, vol. 60, no. 8, pp.
  3929--3941, 2012.

\bibitem{lu2014overview}
Lu~Lu, Geoffrey~Ye Li, A~Lee Swindlehurst, Alexei Ashikhmin, and Rui Zhang,
\newblock ``An overview of massive mimo: Benefits and challenges,''
\newblock {\em IEEE Journal of Selected Topics in Signal Processing}, vol. 8,
  no. 5, pp. 742--758, 2014.

\bibitem{6951471}
M.~Masood, L.~H. Afify, and T.~Y. Al-Naffouri,
\newblock ``Efficient coordinated recovery of sparse channels in massive
  mimo,''
\newblock {\em IEEE Transactions on Signal Processing}, vol. 63, no. 1, pp.
  104--118, Jan 2015.

\bibitem{benedict1967joint}
T~Benedict and T~Soong,
\newblock ``The joint estimation of signal and noise from the sum envelope,''
\newblock {\em IEEE Transactions on Information Theory}, vol. 13, no. 3, pp.
  447--454, 1967.

\bibitem{matzner1994snr}
Rolf Matzner and Ferdinand Englberger,
\newblock ``An snr estimation algorithm using fourth-order moments,''
\newblock in {\em Information Theory, 1994. Proceedings., 1994 IEEE
  International Symposium on}. IEEE, 1994, p. 119.

\bibitem{pauluzzi2000comparison}
David~R Pauluzzi and Norman~C Beaulieu,
\newblock ``A comparison of snr estimation techniques for the awgn channel,''
\newblock {\em IEEE Transactions on communications}, vol. 48, no. 10, pp.
  1681--1691, 2000.

\bibitem{vartiainen2006blind}
Johanna Vartiainen, Harri Saarnisaari, Janne~J Lehtomaki, and Markku Juntti,
\newblock ``A blind signal localization and snr estimation method,''
\newblock in {\em Military Communications Conference, 2006. MILCOM 2006. IEEE}.
  IEEE, 2006, pp. 1--7.

\bibitem{4608971}
W.~Hachem, O.~Khorunzhiy, P.~Loubaton, J.~Najim, and L.~Pastur,
\newblock ``A new approach for mutual information analysis of large dimensional
  multi-antenna channels,''
\newblock {\em IEEE Transactions on Information Theory}, vol. 54, no. 9, pp.
  3987--4004, Sept 2008.

\bibitem{klenke2013probability}
Achim Klenke,
\newblock {\em Probability theory: a comprehensive course},
\newblock Springer Science \& Business Media, 2013.

\bibitem{kay2013fundamentals}
Steven~M Kay,
\newblock {\em Fundamentals of statistical signal processing: Practical
  algorithm development}, vol.~3,
\newblock Pearson Education, 2013.

\bibitem{alfano2004capacity}
Giuseppa Alfano, Antonia~M Tulino, Angel Lozano, and Sergio Verd{\'u},
\newblock ``Capacity of mimo channels with one-sided correlation,''
\newblock in {\em Spread Spectrum Techniques and Applications, 2004 IEEE Eighth
  International Symposium on}. IEEE, pp. 515--519.

\bibitem{shin2006capacity}
Hyundong Shin, Moe~Z Win, Jae~Hong Lee, and Marco Chiani,
\newblock ``On the capacity of doubly correlated mimo channels,''
\newblock {\em IEEE Transactions on Wireless Communications}, vol. 5, no. 8,
  pp. 2253--2265, 2006.

\end{thebibliography}

\end{document}